\documentclass[aps,nopacs,nokeys,superscriptaddress,11pt,twoside,notitlepage]{revtex4-1}

\usepackage{graphicx,epic,eepic,epsfig,amsmath,latexsym,amssymb,verbatim,color}

\usepackage{txfonts}
\usepackage[usenames,dvipsnames]{pstricks}

\usepackage{array}
\usepackage{color}

\newtheorem{definition}{Definition}
\newtheorem{proposition}[definition]{Proposition}
\newtheorem{lemma}[definition]{Lemma}

\newtheorem{theorem}[definition]{Theorem}

\newtheorem{remark}[definition]{Remark}

\def\squareforqed{\hbox{\rlap{$\sqcap$}$\sqcup$}}
\def\qed{\ifmmode\squareforqed\else{\unskip\nobreak\hfil
\penalty50\hskip1em\null\nobreak\hfil\squareforqed
\parfillskip=0pt\finalhyphendemerits=0\endgraf}\fi}
\def\endenv{\ifmmode\;\else{\unskip\nobreak\hfil
\penalty50\hskip1em\null\nobreak\hfil\;
\parfillskip=0pt\finalhyphendemerits=0\endgraf}\fi}
\newenvironment{proof}{\noindent \textbf{{Proof~} }}{\qed}

\mathchardef\ordinarycolon\mathcode`\:
\mathcode`\:=\string"8000
\def\vcentcolon{\mathrel{\mathop\ordinarycolon}}
\begingroup \catcode`\:=\active
  \lowercase{\endgroup
  \let :\vcentcolon
  }

\newcommand{\nc}{\newcommand}
\nc{\rnc}{\renewcommand}

\nc{\ox}{\otimes}
\nc{\dg}{\dagger}
\nc{\dn}{\downarrow}
\nc{\cA}{{\cal A}}
\nc{\cB}{{\cal B}}
\nc{\cC}{{\cal C}}
\nc{\cD}{{\cal D}}
\nc{\cE}{{\cal E}}
\nc{\cF}{{\cal F}}
\nc{\cG}{{\cal G}}
\nc{\cH}{{\cal H}}
\nc{\cI}{{\cal I}}
\nc{\cJ}{{\cal J}}
\nc{\cK}{{\cal K}}
\nc{\cL}{{\cal L}}
\nc{\cM}{{\cal M}}
\nc{\cN}{{\cal N}}
\nc{\cO}{{\cal O}}
\nc{\cP}{{\cal P}}
\nc{\cR}{{\cal R}}
\nc{\cS}{{\cal S}}
\nc{\cT}{{\cal T}}
\nc{\cU}{{\cal U}}
\nc{\cX}{{\cal X}}
\nc{\cZ}{{\cal Z}}
\nc{\csupp}{{\operatorname{csupp}}}
\nc{\qsupp}{{\operatorname{qsupp}}}
\nc{\var}{{\operatorname{var}}}
\nc{\rar}{\rightarrow}
\nc{\lrar}{\longrightarrow}
\nc{\polylog}{{\operatorname{polylog}}}
\nc{\1}{I}
\nc{\wt}{{\operatorname{wt}}}
\nc{\av}[1]{{\left\langle {#1} \right\rangle}}

\def\a{\alpha}
\def\b{\beta}

\nc{\RR}{{{\mathbb R}}}
\nc{\CC}{{{\mathbb C}}}
\nc{\FF}{{{\mathbb F}}}
\nc{\NN}{{{\mathbb N}}}
\nc{\ZZ}{{{\mathbb Z}}}
\nc{\PP}{{{\mathbb P}}}
\nc{\QQ}{{{\mathbb Q}}}
\nc{\UU}{{{\mathbb U}}}
\nc{\EE}{{{\mathbb E}}}
\nc{\id}{{\operatorname{id}}}

\nc{\CHSH}{{\operatorname{CHSH}}}

\nc{\be}{\begin{equation}}
\nc{\ee}{{\end{equation}}}
\nc{\bea}{\begin{eqnarray}}
\nc{\eea}{\end{eqnarray}}
\nc{\<}{\langle}
\rnc{\>}{\rangle}
\nc{\Hom}[2]{\mbox{Hom}(\CC^{#1},\CC^{#2})}
\nc{\rU}{\mbox{U}}

\nc{\ob}[1]{#1}

\nc{\SEP}{{\text{SEP}}}
\nc{\sep}{{\text{sep}}}
\nc{\LOCC}{{\text{LOCC}}}
\nc{\PPT}{{\text{PPT}}}
\nc{\EXT}{{\text{EXT}}}
\nc{\Sym}{{\operatorname{Sym}}}

\newcommand{\bes}{\begin{equation*}}
\newcommand{\ees}{\end{equation*}}

\def\hil{\mathcal{H}}
\def\kil{\mathcal{K}}
\def\bC{\mathbb{C}}
\def\bz{\left(}
\def\jz{\right)}

\def\bR{\mathbb{R}}
\def\X{\mathcal{X}}

\newcommand{\ds}{\mbox{ }\mbox{ }}
\newcommand{\norm}[1]{\left\| #1\right\|}

\newcommand{\ket}[1]{|#1\rangle}
\nc{\ketbra}[2]{|#1\rangle\!\langle#2|}
\nc{\braket}[2]{\langle#1|#2\rangle}
\nc{\proj}[1]{| #1\rangle\!\langle #1 |}
\newcommand{\diad}[2]{|#1\rangle\langle #2|}
\newcommand{\pr}[1]{\diad{#1}{#1}}

\newcommand{\powerset}{\mathcal{P}}
\newcommand{\powersetnonempty}{\mathcal{P}_{\emptyset}}

\DeclareMathOperator{\Tr}{Tr}
\DeclareMathOperator{\tr}{Tr}

\DeclareMathOperator{\rank}{rank}

\begin{document}

\title{The structure of R\'{e}nyi entropic inequalities}

\author{Noah Linden}
\email{n.linden@bristol.ac.uk}
\affiliation{Department of Mathematics, University of Bristol, Bristol BS8 1TW, U.K.}

\author{Mil\'an Mosonyi}
\email{milan.mosonyi@gmail.com}
\affiliation{Department of Mathematics, University of Bristol, Bristol BS8 1TW, U.K.}
\affiliation{Mathematical Institute, Budapest University of Technology and Economics,
Budapest, Egry J. u.~1, 1111 Hungary}

\author{Andreas Winter}
\email{der.winter@gmail.com}
\affiliation{ICREA \&{}
             F\'{\i}sica Te\`{o}rica: Informaci\'{o} i Fenomens Qu\`{a}ntics,
             Universitat Aut\`{o}noma de Barcelona, ES-08193 Bellaterra (Barcelona), Spain}
\affiliation{Department of Mathematics, University of Bristol, Bristol BS8 1TW, U.K.}
\affiliation{Centre for Quantum Technologies, National University of Singapore, 3 Science Drive 2, Singapore 117543}

\begin{abstract}

 We investigate the universal inequalities relating the $\alpha$-R\'enyi entropies
  of the marginals of a multi-partite quantum state. This is in analogy
  to the same question for the Shannon and von Neumann entropy ($\alpha=1$) which are known to 
 satisfy several non-trivial inequalities such as strong subadditivity.
Somewhat surprisingly, we find for $0<\alpha<1$, that the only inequality is
  non-negativity:
  In other words, any collection of non-negative numbers assigned to the
  nonempty subsets of $n$ parties can be arbitrarily well approximated by the
  $\alpha$-entropies of the $2^n-1$ marginals of a quantum state.

  For $\alpha>1$ we show analogously that there are no non-trivial
  \emph{homogeneous} (in particular no linear) inequalities.
On the other hand, it is known that there are further,
 non-linear and indeed non-homogeneous, inequalities delimiting the
 $\alpha$-entropies of a general quantum state.

  Finally, we also treat the case of R\'enyi entropies restricted to
  classical states (i.e.~probability distributions), 
which in addition to non-negativity are also subject to monotonicity.
For $\alpha \neq 0,1$
  we show that this is the only other homogeneous relation.
\end{abstract}

\date{30 November 2012}

\maketitle

\section{Prologue}
\label{sec:intro}
The von Neumann entropy $S(\rho)=-\tr \rho\log\rho$ of a quantum state $\rho$ is a key 
notion in quantum information theory \cite{q-Shannon} 
as well as in statistical physics \cite{vonNeumann}. It is furthermore 
the canonical measure of entanglement for bipartite pure states \cite{entanglement}. 
In many cases the relative 
magnitude of the entropy of the reduced states of different subsystems is important.
Thus for example for a tripartite state $\rho_{ABC}$ one can compute $S(\rho_{A})$, 
$S(\rho_{B})$, and so on. For any positive number $a$ one can find a quantum state such that 
$S(\rho_{A})=a$, for example. However for a fixed quantum state, there are inequalities 
between the values of the entropies of the reduced states of the subsystems.

There are essentially two such unconstrained inequalities known (up to permutation of the parties),
\emph{strong subadditivity} and \emph{weak monotonicity} \cite{LiebRuskai:SSA,Pippenger:q}:
\begin{equation}\begin{split}
  \label{eq:SSA}
  S(\rho_{ABC}) + S(\rho_B) &\leq S(\rho_{AB}) + S(\rho_{BC}), \\
  S(\rho_{A}) + S(\rho_B)   &\leq S(\rho_{AC}) + S(\rho_{BC}).
\end{split}\end{equation}

There are no other constraints for up to $3$ parties, but the analogous
statement is a major open problem for larger 
$n$~\cite{Pippenger:q,LW:new,CadneyLindenWinter}.   Indeed we anticipate that the question 
may be very complicated in general.  For example the analogous question for classical 
(Shannon) entropies has been much studied and in this case, for $n\ge 4$ there are 
infinitely many independent (linear) inequalities known~\cite{ZY-1,ZY-2,Matus}. All these 
inequalities, and
more discovered by Makarychev \emph{et al.}~\cite{Makarychev}
and Dougherty \emph{et al.}~\cite{DFZ} might very well hold for the
von Neumann entropy, too.
 Indeed it is known that the von Neumann entropy satisfies some constrained inequalities that are counterparts of known classical constrained inequalties \cite{CadneyLindenWinter}.

For a state of $n$ parties there are $2^n-1$ non-trivial von Neumann entropies,
one corresponding to each non-empty subset of parties.
Thus the existence of these inequalities means that given a set of $2^n-1$ positive numbers 
there will, in general, be no quantum state whose reduced state entropies have these values.

In this paper we consider the analogous questions for quantum R\'enyi entropies, also called 
$\alpha$-entropies~\cite{Renyi:entropy},
of a quantum state $\rho$ (given as a unit trace
density operator on a suitable Hilbert space):
\begin{equation}
  S_\alpha(\rho) = \frac{1}{1-\alpha} \log\tr\rho^\alpha,\label{eq:renyi}
\end{equation}
for $0 \leq \alpha \leq \infty$. 
(We note that the case of $\alpha=1$ is the von Neumann entropy).

We will show that, very surprisingly, the case of R\'enyi entropies for $\alpha\neq 1$ is 
much different from that for the von Neumann entropy:

For $0 < \alpha < 1$, we will show 
(in theorem~\ref{thm:0-alpha-1})
that the only inequality is non-negativity 
$S_\alpha(\rho)\geq 0$. In other words, any collection of non-negative numbers assigned to 
the nonempty subsets of $n$ parties can be arbitrarily well approximated by the
$\alpha$-entropies of the $2^n-1$ marginals of a quantum state.
  
For $\alpha > 1$ we will prove 
that there are no linear (or indeed homogeneous) inequalities. 
We show (in theorem~\ref{thm:1-alpha-infty})
that given any vector ${\bf v}$ of $2^n-1$ positive numbers, it may or may not 
be the case that this is the vector of the $2^n-1$ marginals of a quantum state; however 
it is arbitrarily well approximated by a positive multiple of
the $\alpha$-entropies of the $2^n-1$ marginals of a quantum state.
On the other hand (contrary to the case $0<\alpha<1$) there are other, 
nonlinear, inequalities delimiting the set of possible entropy vectors; one such inequality 
was proved in \cite{Audenaert:inequality}, which we recall in 
section~\ref{sec:1-alpha-infty}.

Finally, we show (in section~\ref{sec:classical}) that in the classical case
the only homogeneous inequalities
are non-negativity and monotonicity (under the inclusion of subsets of parties), for all $\alpha \neq 0,1$.

\section{The R\'enyi entropy}
\label{sec:renyi}
The definition in (\ref{eq:renyi}) is
clearly well-defined, and continuous in the state as well as in
$\alpha$, for $\alpha \in (0,\infty)\setminus\{1\}$.
For $\alpha = 0,1,\infty$, the function is defined by taking a
limit, yielding
\begin{align*}
  S_0(\rho)      &= \log \rank \rho, \\
  S_1(\rho)      &= S(\rho) := -\tr\rho\log\rho \ \text{ (von Neumann entropy)}, \\
  S_\infty(\rho) &= -\log \| \rho \|,
\end{align*}
where $\|\cdot\|$ denotes the operator norm, i.e.~$\|\rho\|$ is the
largest eigenvalue of $\rho$.

By their definition, all of the quantum R\'enyi entropies depend only on the
spectrum of $\rho$, which we can think of as a probability distribution $P$.
In this sense, 
the above formulas generalise the notion
\[
  H_\alpha(P) = \frac{1}{1-\alpha} \log \left( \sum_x P(x)^\alpha \right)
\]
introduced by R\'enyi in his axiomatic investigation
of information measures for random variables and their
distributions~\cite{Renyi:entropy}, following Shannon's example
\cite{Shannon}. This approach has generated a lot of subsequent activity
\cite{AczelDaroczy}.

It is easy to see that for states $\rho \geq 0$, $\tr\rho = 1$,
all $S_\alpha(\rho) \geq 0$, with equality if and only if $\rho$
is pure, i.e.~a rank-one projector $\rho = \proj{\varphi}$.
Furthermore, for fixed $\rho$, the function $\alpha \mapsto S_\alpha(\rho)$
is monotonically non-increasing~\cite{Renyi:entropy}.

Many other useful, interesting and curious mathematical properties of
the R\'enyi entropies are known~\cite{AczelDaroczy}. 

R\'enyi entropies, and, more generally, R\'enyi relative entropies and the corresponding channel capacities play an important role in classical as well as quantum information theory.
The R\'enyi quantities with parameter $\alpha\in(0,1)$ are related to the so-called direct domain of information theoretic problems.
They can be used to quantify the trade-off between the rates of the two types of error probabilities in binary state discrimination \cite{Csiszar,Hayashi,Nagaoka,HMO,ANSzV}, which in turn yields a trade-off relation between the error rate and the compression rate in state compression
(see \cite{Csiszar} for the classical case; the quantum case is completely analogous). The related capacities
quantify the trade-off between the error rate and the coding rate for classical information
transmission \cite{Csiszar,Hayashi}, and can be used to obtain
lower bounds on the single-shot classical capacities \cite{MD}.
The R\'enyi quantities with parameter $\alpha>1$ are related to converse problems. They can be used to quantify the trade-off between the rates of the type I error and the type II success probability in binary state discrimination \cite{Csiszar,ON,Hayashibook}, as well as the trade-off between the
rate of the success probability and the compression rate in state compression \cite{Csiszar}. The related capacity formulas give bounds on the success rate for coding rates above the Holevo capacity
\cite{Csiszar,ON,KW}, and can be used to give upper bounds on the single-shot classical capacities of quantum channels \cite{MH}.
Also, R\'{e}nyi entropies feature prominently in the theory of
bipartite pure state transformations by local operations and classical
communication: Only recently it was shown~\cite{AubrunNechita} that
the monotonicity of the R\'enyi entropies of the reduced states for
$\alpha > 1$ is both necessary and sufficient for catalytic transformations
(whereas unassisted transformations are long known to be characterized by
majorization~\cite{Hardy99,Nielsen99}). And in~\cite{HaydenWinter:c-cost},
R\'{e}nyi entropies (essentially $\alpha=0$ and $\alpha=\infty$) were used
to put bounds on the classical communication required for a given
transformation.
And finally, R\'{e}nyi entropies were employed to put lower bounds on
the communication complexity of certain distributed computation problems
\cite{vanDamHayden,MontanaroWinter}.

While the von Neumann entropy can be obtained as the limit of the R\'enyi entropies for 
$\alpha\to 1$, and hence it can be considered as one particular member of this parametric family of entropies, its basic properties sharply distinguish it from all other members of the family. 
Indeed, while the von Neumann entropy is strongly subadditive, the other R\'enyi entropies with 
$\alpha\in (0,+\infty)\setminus\{1\}$ are not even subadditive. 

To illustrate the consequences of this difference, we mention 
the problem of entropy asymptotics on spin chains.
Given a 
translation-invariant state $\rho$ on an infinite spin chain, subadditivity of entropy 
ensures the existence of the limit 
$s(\rho):=\lim_{n\to+\infty}\frac{1}{n}S(\rho_{[1,n]})$, where $\rho_{[1,n]}$ is the 
restriction of $\rho$ to any $n$ consecutive sites, and $s(\rho)$ gives the ultimate 
compression rate for an ergodic $\rho$ \cite{BSz}. More refined knowledge about the decay of 
error for rates below $s(\rho)$ can be obtained using the method developed in \cite{HMO};
for this, however, one has to show the existence of the regularized R\'enyi entropies
$s_\alpha(\rho):=\lim_{n\to+\infty}\frac{1}{n}S_\alpha(\rho_{[1,n]})$ for every 
$\alpha\in(0,1)$. Due to the lack of subadditivity, the existence of this limit is 
not at all straightforward, and is actually only known for some special classes of states
\cite{HMO,MHOF,M}.

When $\rho$ is pure, the block entropies $S_\alpha(\rho_{[1,n]})$ are used to 
measure the entanglement between the block $[1,n]$ and the rest of the chain,
and the scaling of these entropies are closely related to the presence or absence of 
criticality in the system \cite{VLRK,ECP}.
It follows from strong subadditivity that the entanglement entropy
$S(\rho_{[1,n]})$ is a monotone increasing function of the block size
\cite{AF}. This is no longer true when the entanglement is measured by some R\'enyi entropy; a 
counterexample with oscillating block R\'enyi entropies for $\alpha>2$ was found in \cite{Illuminati}. It 
is not known, however, whether such oscillating behaviour can happen for R\'enyi entopies 
with parameter $\alpha$ arbitrarily close to $1$. 

In the view of the above examples, it is natural to ask whether there are other 
universal inequalities between the R\'enyi entropies of the subsytems of a multipartite 
quantum system, and this is what we are going to investigate in the following.

To fix notation, we shall concern ourselves with $n$-partite quantum systems
with generic tensor product Hilbert space $\cH = \cH_1 \otimes \cdots \otimes \cH_n$.
Within the discussion we usually consider $n$ and $\alpha$ to be fixed,
but the local systems are unconstrained,
i.e.~we do not impose limits on the dimension of the $\cH_i$.
For a state $\rho$ on $\cH$ we have the reduced states
$\rho_I = \tr_{I^c} \rho$, with the partial trace over all parties
in the complement $I^c = [n]\setminus I$, and we shall consider them
and their entropies all at once, for all non-empty subsets
of $[n]=\{1,\ldots,n\}$.
The power set and the power set without the empty set we denote as follows:
\begin{align*}
  \powerset[n]         &:= \{ I \subset [n] \}, \\
  \powersetnonempty[n] &:= \{ I \subset [n] : I \neq \emptyset \}
                         = \powerset[n] \setminus \{\emptyset\}.
\end{align*}

We are interested in the universal relations obeyed by the
$\alpha$-entropies of a general $n$-party state $\rho$. For instance,
by definition clearly 
\begin{equation*}
S_\alpha(\rho_I) \geq 0
\end{equation*}
for all subsets $I\subset [n]$.

Note that via the usual diagonal matrix representation we can view a
probability distribution of $n$ discrete random variables as a quantum state, and
conversely, states which are diagonal in a tensor product basis of an $n$-party system
can be identified with a classical $n$-party probability distribution, and hence 
we will call states of that form \emph{classical}.
In this case, there is another inequality,
\begin{equation*}
S_\alpha(\rho_I) \leq S_\alpha(\rho_J)
\end{equation*} 
for $I\subset J$, i.e.~monotonicity
of the entropy function with respect to subset inclusion.

These examples motivate the introduction of the set of all entropic vectors,
\[
  \Sigma_\alpha^n := \left\{ \bigl( S_\a(\rho_I) \bigr)_{I\in\powersetnonempty[n]}
                                                          \,:\, \rho \text{ state} \right\}
                  \subset \RR^{\powersetnonempty[n]},
\]
and the same for classical case
\[\begin{split}
  \Gamma_\alpha^n &:= \left\{ \bigl( S_\a(\rho_I) \bigr)_{I\in\powersetnonempty[n]}
                                                 \,:\, \rho \text{ classical state} \right\} \\
                  &=  \left\{ \bigl( H_\a(P_I) \bigr)_{I\in\powersetnonempty[n]}
                                                       \,:\, P \text{ prob. distr.} \right\}
                   \subset \RR^{\powersetnonempty[n]}.
\end{split}\]
In fact, as we tend to consider ``$\leq$''type inequalities
between continuous functions of the coordinates, it makes sense
to focus on the topological closures
$\overline{\Sigma_\alpha^n},\ \overline{\Gamma_\alpha^n} \subset \RR^{\powersetnonempty[n]}$.
Those universal inequalities we are looking for are the constraints
describing the geometric shape of these sets.

The above examples of known inequalities are homogeneous,
indeed linear, relations. (By a homogeneous inequality we mean an inequality of the 
form $f({\bf v})\ge 0,\,{\bf v}\in \Sigma_\alpha^n$, where $f$ is a homogeneous function
on $\RR^{\powersetnonempty[n]}_{\ge 0}$, i.e., there exists a $d\in\bR$ such that  
$f(\lambda{\bf v})=\lambda^d f({\bf v})$ holds for every $\lambda\in\bR_{\ge 0}$ and 
${\bf v}\in \RR^{\powersetnonempty[n]}_{\ge 0}$.) That it is meaningful to look for such
relations is motivated by the observation that all R\'enyi entropies
are \emph{extensive}, i.e.
\[
  S_\alpha(\rho\otimes\sigma) = S_\alpha(\rho) + S_\alpha(\sigma).
\]
And since this is true for all subset reduced states simultaneously,
we have for non-negative integers $k$ and $\ell$,
\[
  k\Sigma_\alpha^n + \ell\Sigma_\alpha^n \subset \Sigma_\alpha^n,
    \quad
  k\Gamma_\alpha^n + \ell\Gamma_\alpha^n \subset \Gamma_\alpha^n,
\]
and likewise for the respective closures. If this held for non-negative
\emph{reals} it would mean that the corresponding set is a convex
cone. This is indeed known for $\alpha=1$~\cite{Yeung:framework,Pippenger:q},
but not true for $\alpha > 1$ (see below).

\section{$\mathbf{0 < \alpha < 1}$}
\label{sec:0-alpha-1}
In this section, $\alpha$ is fixed in the interval $(0,1)$. We start off
with a simple classical construction.
For $I\in\powersetnonempty[n]$, let $\delta_I$ denote the corresponding basis vector in 
$\RR^{\powersetnonempty[n]}$, i.e., $\delta_I$ is the characteristic function of the singleton $\{I\}$.

\begin{lemma}
  \label{lemma:0-alpha-1}
  For any $s > 0$, the vector $s\delta_{[n]} \in \RR_{\geq 0}^{\powersetnonempty[n]}$
  is approximately $\alpha$-entropic, i.e.~$s\delta_{[n]} \in \overline{\Sigma_\alpha^n}$.
  In fact, this vector can be approximated arbitrarily by classical states.
\end{lemma}
\begin{proof}
  For integers $M_1,\ldots,M_n$ consider ``local'' alphabets
  $\cX_i := \{ 0 \} \cup [M_i]$ and define distributions $P_{t;\{M_i\}}$
  ($0\leq t\leq 1$)
  on the Cartesian product $\cX_1\times\cdots\times\cX_n$ as follows:
  \[
    P_{t;\{M_i\}}(x_1,\ldots,x_n) := \begin{cases}
                                       1-t                     & \text{ if } x_1=\ldots=x_n=0,      \\
                                       \frac{t}{M_1\cdots M_n} & \text{ if } x_1,\ldots,x_n \neq 0, \\
                                       0                       & \text{ otherwise.}
                                     \end{cases}
  \]
  The marginals on $\cX_I = \varprod_{i\in I} \cX_i$, for a subset $I\subset [n]$,
  are easy to construct: they are given precisely by
  \(
    P_{t;\{M_i:i\in I\}}(x_i:i\in I)
  \).
  The corresponding quantum state and its marginals hence are
  \begin{align*}
    \rho   &= \sum_{x_1,\ldots,x_n} P_{t;\{M_i\}}(x_1,\ldots,x_n) \proj{x_1} \ox \cdots \ox \proj{x_n}, \\
    \rho_I &= \sum_{x_i:i\in I} P_{t;\{M_i:i\in I\}}(x_i:i\in I) \bigotimes_{i\in I} \proj{x_i}.
  \end{align*}
  With this, the R\'enyi entropies are straightforward to compute:
  \[
    S_\alpha(\rho_I) = \frac{1}{1-\alpha}\log\left( (1-t)^\alpha + t^\alpha M_I^{1-\alpha} \right),
  \]
  with $M_I = \prod_{i\in I} M_i$.

  Now, for sufficiently large $M_{[n]}$, we can set
  \begin{equation}
    \label{eq:t-value}
    t := \left( \frac{2^{s(1-\alpha)}-1}{M_{[n]}^{1-\alpha}} \right)^{\frac{1}{\alpha}}.
  \end{equation}
  Then, in the limit $\min\{M_1,\ldots,M_n\} \rightarrow \infty$,
  \[
    S_\alpha(\rho_{[n]}) = \frac{1}{1-\alpha}\log\left( (1-t)^\alpha + 2^{s(1-\alpha)} - 1 \right)
                           \rightarrow s,
  \]
  since $t \rightarrow 0$.
  On the other hand, for $I \subsetneq [n]$,
  \[
    S_\alpha(\rho_I) = \frac{1}{1-\alpha}
                       \log\left( (1-t)^\alpha + \frac{2^{s(1-\alpha)}-1}{M_{[n]\setminus I}^{1-\alpha}} \right)
                       \rightarrow 0,
  \]
  because $M_i \rightarrow \infty$.
\end{proof}

\begin{proposition}
  \label{prop:0-alpha-1}
  For any $\emptyset\neq I\subset[n]$, and $s > 0$, the
  vector $s\delta_I \in \RR_{\geq 0}^{\powersetnonempty[n]}$
  is approximately $\alpha$-entropic, i.e.~$s\delta_I \in \overline{\Sigma_\alpha^n}$.
\end{proposition}
\begin{proof}
  It is enough to show that for any $\epsilon > 0$, there exist local
  systems $\cH_1,\ldots,\cH_n,\cH_{n+1}$ and a pure state $\rho = \proj{\psi}$ on
  $\cH_1\ox\cdots\ox\cH_n\ox\cH_{n+1}$ with
  \[
    s - \epsilon \leq S_\alpha(\rho_I) \leq s + \epsilon, \quad \text{and} \quad
    S_{\alpha}(\rho_J)\leq \epsilon \ \text{ if } J\subset[n],\ J\neq I.
  \]

  If $I=[n]$, 
we just use the $(n+1)$st party to purify the state. If
  $|I| = 1$, we likewise take the classical state of lemma~\ref{lemma:0-alpha-1}
  on the $n$-party system $[n+1]\setminus I$ and purify it using the
  system $I$.
  Thus, from now on we may assume that $k=|I|$ and $\ell = |I^c| = n+1-k$
  are both $\geq 2$. The idea is that for integer $M$, the distributions
  $P_{t;\{M^\ell:i\in I\}}$ on the systems $I$,
  and $P_{t;\{M^k:j\in I^c\}}$
  on the systems $I^c = [n+1]\setminus I$, both with the same $t$ given
  by 
  \[
    t = \left( \frac{2^{s(1-\alpha)}-1}{M^{k\ell(1-\alpha)}} \right)^{\frac{1}{\alpha}},
  \]
  have the same nonzero probabilities, just
  arranged differently. In other words, the corresponding classical states
  are isospectral, hence we may view them as reduced states of an
  $(n+1)$-party pure state.

  In detail, we may without loss of generality relabel the systems such that
  $I = \{1,\ldots,k\}$ and $I^c = \{k+1,\ldots,k+\ell=n+1\}$.
For every $i\in\{1,\ldots,k\}$ and $j\in\{1,\ldots,\ell\}$, let $\kil_{ij}$ be an
$M$-dimensional Hilbert space with an orthonormal basis $\{\ket{k}_{ij}\,:\,k=1,\ldots,M\}$,
and define
\begin{equation*}
\hil_i:=\bC\ket{0}_i\oplus\bz\bigotimes_{j=1}^{\ell}\kil_{ij}\jz,\ds\ds\ds
\hil_{k+j}:=\bC\ket{0}_{k+j}\oplus\bz\bigotimes_{i=1}^k\kil_{ij}\jz,
\end{equation*}
where $\ket{0}_i$ are unit vectors.
For a $k\times \ell$ matrix $x\in[M]^{k\times \ell}$ and $i\in\{1,\ldots,k\},\,
j\in\{1,\ldots,\ell\}$, let
\begin{equation*}
\ket{x^{(i)}}:=\bigotimes_{j=1}^{\ell}\ket{x^{(i)}_j}_{ij}\in\hil_i,\ds\ds\ds
\ket{x_{(k+j)}}:=\bigotimes_{i=1}^k\ket{x^{(i)}_j}_{ij}\in\hil_{k+j},
\end{equation*}
and
\begin{equation*}
\ket{\psi}:=\sqrt{1-t}\bigotimes_{i=1}^{n+1}\ket{0}_i+
\sqrt{\frac{t}{M^{k\ell}}}\sum_{x\in[M]^{k\times \ell}}
\ket{x^{(1)}}\ldots\ket{x^{(k)}}\ket{x_{(k+1)}}\ldots\ket{x_{(n+1)}}.
\end{equation*}
(Here $\ket{x}\ket{y}$ stands for $\ket{x}\otimes\ket{y}$.)

Let $\rho := \rho_{[n]} = \tr_{n+1} \proj{\psi}$.
The crucial property of this definition is that every party $i\in I$
and $j\in I^c$ have a coordinate in common, namely $x^{(i)}_j \in [M]$.
One can easily see that $\rho_I$ is a classical state of the type studied in
lemma \ref{lemma:0-alpha-1}, and the same calculation as in lemma \ref{lemma:0-alpha-1}
shows that
\begin{equation*}
S_\alpha(\rho_I) =
\frac{1}{1-\alpha}\log\left( (1-t)^\alpha + 2^{s(1-\alpha)} - 1 \right) \rightarrow s,
\end{equation*}
as $M\to\infty$.
If $J\in\powersetnonempty[n]$ is different from $I$ then there are
  $i\in I$ and $j\in I^c = [n+1]\setminus I$, such that either $i,j\in J$
  or $i,j\in J^c$. The second case has entropy equivalent to the first
  since we may just go to the complementary set. In the first case,
we have
\begin{equation*}
\rho_J=(1-t)\bigotimes_{i\in J}\pr{0}_i+t\sigma,
\end{equation*}
where $\sigma$ is supported on a space which contains each $\kil_{ij}$
at most once
if the $i$th or the $(k+j)$th system has been traced out, and twice
otherwise. Hence,
\begin{equation*}
\sigma=\bz\bigotimes_{i\in J\cap I}\bigotimes_{j\in J\cap [n]\setminus I}\pr{\psi_{ij}}\jz\otimes\sigma',
\end{equation*}
where
\begin{equation*}
\ket{\psi_{ij}}=\frac{1}{\sqrt{M}}\sum_{k=1}^M\ket{k}_{ij}\ket{k}_{ij}\in\kil_{ij}\otimes\kil_{ij},
\end{equation*}
and $\sigma'$ is a density operator supported on a space of dimension at most
$M^{k\ell-1}$. Hence, $S_{\alpha}(\sigma)=S_{\alpha}(\sigma')\le\log M^{k\ell-1}$, or equivalently, $\Tr \sigma^{\alpha}\le M^{(k\ell-1)(1-\alpha)}$. This yields
\begin{equation*}\begin{split}
S_\alpha(\rho_J)
&\leq
\frac{1}{1-\alpha}\log\left( (1-t)^\alpha + M^{(k\ell-1)(1-\alpha)}t^\alpha \right) \\
&=
\frac{1}{1-\alpha}\log\left( (1-t)^\alpha + \frac{1}{M^{1-\alpha}}
\bigl(2^{s(1-\alpha)}-1 \bigr) \right)
\rightarrow 0,
  \end{split}
\end{equation*}
as $M\to\infty$.
\end{proof}

\begin{theorem}
  \label{thm:0-alpha-1}
  Every element $v \in \RR_{\geq 0}^{\powersetnonempty[n]}$ is approximately
  $\alpha$-entropic.
  In other words, there are no non-trivial
  inequalities constraining the R\'enyi entropies (with fixed $\alpha<1$)
  of a multi-party state: The only restriction is non-negativity:
  \[
    \overline{\Sigma_\alpha^n} = \RR_{\geq 0}^{\powersetnonempty[n]}.
  \]
\end{theorem}
\begin{proof}
  Via proposition~\ref{prop:0-alpha-1} this is quite obvious:
  Observe $v = \sum_{I\in\powersetnonempty[n]} v_I \delta_I$, and that
  for each subset $I\subset [n]$ we can find an $n$-party state $\rho^{(I)}$
  such that its entropies arbitrarily approximate $v_I\delta_I$, i.e.
  \[
    \bigl| S_\alpha(\rho^{(I)}_J) - v_I\delta_{I,J} \bigr|
                           \leq \epsilon \ \text{ for all } I,J \subset [n].
  \]
  Letting $\rho := \bigotimes_{I\in\powersetnonempty[n]} \rho^{(I)}$,
  we are done, since $\epsilon$ can be chosen arbitrarily small.
\end{proof}

\begin{remark}
  \label{rem:0-1-closure}
  From theorem~\ref{thm:0-alpha-1} we can see that $\Sigma_\alpha^n$ is not a closed set
  (assuming $n\geq 2$).
  \normalfont
  Indeed, we found that for any $I \subset [n]$, the ray $\RR_{\geq 0}\delta_I$
  is in $\overline{\Sigma_\alpha^n}$. However, it is easy to see that except
  for the origin, none of its points $s\delta_I$ can be an element of $\Sigma_\alpha^n$.

  For otherwise there would be a state $\rho$ with $S_\alpha(\rho_I) = s$ and all
  other $S_\alpha(\rho_J)=0$. Now, if $|I| \geq 2$, say $I=\{i_1,i_2,\ldots\}$,
  then $S_\alpha(\rho_{i})=0$ implies that all single-party marginals
  $\rho_{i}$ are pure, meaning that $\rho = \rho_1 \ox \cdots \ox \rho_n$
  is pure, too. Hence we would have $S(\rho_I)=0$ as well.
  If on the other hand $I=\{i\}$, then we may choose $j\not\in I$ and
  reason similarly that $\rho_j$ is
  a pure state, hence $\rho_{\{i,j\}} = \rho_i \ox \rho_j$ and so
  $S(\rho_{\{i,j\}}) = S(\rho_i) + S(\rho_j) = s \neq 0$, obtaining
  a contradiction again.
\end{remark}

\section{$\mathbf{1 < \alpha \leq \infty}$}
\label{sec:1-alpha-infty}
As in the previous section, we start with the basic construction
to attain entropy vectors arbitrarily close to the coordinate axes.
Throughout the section, $1 < \alpha \leq \infty$ is fixed.

\begin{lemma}
  \label{lemma:1-alpha-infty}
  For all $s > 0$, there is a vector
  $s\delta_{[n]} + O(1) \in \Sigma_\alpha^n$.
  To be precise, there exists a constant $C$, which may be chosen
  as $C=\frac{1}{1-\frac{1}{\alpha}}\log n$, and classical states with
  \[
    s \leq S_\a(\rho) \leq s + C, \quad \text{and} \quad
    S_a(\rho_J) \leq C \ \text{ for } J\neq [n].
  \]
  In particular, $\delta_{[n]} \in \overline{\RR_{\geq 0}\,\Sigma_\alpha^n}$.
\end{lemma}
\begin{proof}
  The following argument is presented for $\alpha < \infty$; to obtain
  the claims in the case $\alpha = \infty$ one simply takes the limit.

  For an integer $M$ consider ``local'' alphabets
  $\cX_i := \{ 0 \} \cup [M]$ and define distributions $Q_{[n]:R}$
  on the Cartesian product $\cX_1\times\cdots\times\cX_n$ as follows:
  \[
    Q_{[n]:R}(x_1,\ldots,x_n)
        := \begin{cases}
             \frac1n R(x_i) & \text{ if } x_i \neq 0 \text{ and } x_j = 0 \forall j\neq i, \\
             0              & \text{ otherwise},
           \end{cases}
  \]
  where $R$ is an arbitrary probability distribution on $[M]$.

  For the corresponding classical state, it is straightforward to verify that
  \[
    S_\alpha(\rho_{[n]})
         = \frac{1}{1-\alpha}\log\left( n \sum_{x=1}^M \left(\frac1n R(x)\right)^\alpha \right)
         = \log n + H_\alpha(R).
  \]
  On the other hand, the marginal state $\rho_I$ for any $I \subsetneq [n]$ has
  an eigenvalue $\lambda\geq \frac1n$, hence
  \begin{equation*}
    S_\alpha(\rho_I) \leq 
\frac{1}{1-\alpha}\log \lambda^{\alpha}
\le
\frac{1}{1-\alpha}\log \left(\frac1n\right)^\alpha = C,
  \end{equation*}
  and we are done, since we can choose $M$ large enough to accommodate
  a distribution $R$ on $[M]$ with $H_\alpha(R) = s$.
\end{proof}

\begin{proposition}
  \label{prop:1-alpha-infty}
  For all $s > 0$ and $I\in\powersetnonempty[n]$, there is a vector
  $s\delta_I + O(1) \in \Sigma_\alpha^n$.
  To be precise, there exists a constant $C$, which may be chosen
  as $C=\frac{1}{1-\frac{1}{\alpha}}\log(|I|(n+1-|I|))$, and states with
  \[
    s \leq S_\a(\rho_I) \leq s + C, \quad \text{and} \quad
    S_a(\rho_J) \leq C \ \text{ for } J\neq I.
  \]
  In particular, $\delta_I \in \overline{\RR_{\geq 0}\,\Sigma_\alpha^n}$.
\end{proposition}
\begin{proof}
  If $I=[n]$, this is lemma~\ref{lemma:1-alpha-infty}, but we shall 
  present a direct quantum construction, of a pure state on $n+1$ parties;
  hence view $I$ as a subset of $[n+1]$. Pick a distribution $R$ on
  some finite alphabet $[M]$ with $H_\alpha(R) = s$ and fix a purification
  of $R$ -- understood as a quantum state $R = \sum_x R(x) \proj{x}$ --,
  $\ket{\mu} = \sum_{x=1}^M \sqrt{R(x)} \ket{x}\ket{x}$.
  
  Now, construct the following $(n+1)$-party pure state vector,
  \[
    \ket{\psi} := \sqrt{\frac{1}{|I||I^c|}} 
                      \bigoplus_{i\in I,\, j\in I^c} \ket{\mu}_{ij} \ox \ket{ij}^{\ox [n]},
  \]
  where $I^c = [n+1]\setminus I$, and $\rho := \tr_{n+1} \proj{\psi}$. 
  However, our reasoning will be based on the pure state $\proj{\psi}$.
  Above, the direct sum means that we take direct sums of the
  local Hilbert spaces, which we indicate by the label ``$ij$'' attached
  to each local system, whereas $\ket{\mu}_{ij}$ is the state shared between
  parties $i$ and $j$.
  
  It is straightforward to check that
  \[
    \rho_I = \frac{1}{|I||I^c|} \bigoplus_{i\in I,\, j\in I^c} R_{ij} \ox \proj{ij}^{I},
  \]
  hence $S_\alpha(I) = s + \log(|I||I^c|)$. On the other hand, if $J \subset [n]$
  with $J\neq I$, then there exist $i\in I$ and $j\in I^c$ such that either both $i,j \in J$ or
  both $i,j \in J^c$. Thus, $\rho_J$ has a direct sum component $\proj{\mu}$
  and as a consequence an eigenvalue $\geq \frac{1}{|I||I^c|}$, hence
  \[
    S_\alpha(J) \leq \frac{1}{1-\alpha}\log \left(\frac{1}{|I||I^c|}\right)^\alpha = C,
  \]
  and we are done.
\end{proof}

\begin{theorem}
  \label{thm:1-alpha-infty}
  For every element $v \in \RR_{\geq 0}^{\powersetnonempty[n]}$ there is
  a vector $v + O(1) \in \Sigma_\alpha^n$. To be precise, there exists a
  constant $C$, which may be chosen as
  $C=\frac{1}{1-\frac{1}{\alpha}}(\log(n+1))2^{n+1}$,
  and states with
  \[
    \bigl| S_\alpha(\rho_I) - v_I \bigr| \leq C \ \text{ for all } I \subset [n].
  \]
  In other words, there are no non-trivial homogeneous
  inequalities constraining the R\'enyi entropies (with fixed $\alpha>1$)
  of a multi-party state: The only restriction is non-negativity:
  \[
    \overline{\RR_{\geq 0}\Sigma_\alpha^n} = \RR_{\geq 0}^{\powersetnonempty[n]}.
  \]
\end{theorem}
\begin{proof}
  Using proposition~\ref{prop:1-alpha-infty} this is trivial:
  $v = \sum_{I\in\powersetnonempty[n]} v_I \delta_I$, and
  for each subset $I\subset [n]$ we can find an $n$-party state $\rho^{(I)}$
  such that its entropies approximate $v_I\delta_I$, i.e.
  \[
    \bigl| S_\alpha(\rho^{(I)}_J) - v_I\delta_{I,J} \bigr|
                \leq \frac{1}{1-\frac{1}{\alpha}}\log n \ \text{ for all } I,J \subset [n].
  \]
  Letting $\rho := \bigotimes_{I\in\powersetnonempty[n]} \rho^{(I)}$,
  we are done.
\end{proof}

\begin{remark}
  \label{rem:1-infty-closure}
  From theorem~\ref{thm:1-alpha-infty} we can see that $\RR_{\geq 0}\Sigma_\alpha^n$
  is not a closed set (assuming $n\geq 2$).
  \normalfont
  This is argued in the same way as in remark~\ref{rem:0-1-closure}
  for the case $\alpha < 1$.
\end{remark}

\medskip
The above theorem~\ref{thm:1-alpha-infty} looks quite similar to
theorem~\ref{thm:0-alpha-1} for $0<\alpha<1$. However, whereas
there we could conclude that there are no nontrivial inequalities
whatsoever for the R\'enyi entropies of a multi-party state, here we only
get that there cannot be any further \emph{homogeneous} inequalities
apart from non-negativity.

That this is the most we can hope to obtain follows from the observation
that there are other, non-linear and non-homogeneous, inequalities
constraining the entropy vectors. In fact, $\overline{\Sigma_\alpha^n}$
is not a cone at all for $\alpha > 1$!

An example of such an inequality was presented by
Audenaert~\cite{Audenaert:inequality}: The $\alpha$-Schatten norm
$\|\rho\|_\alpha = \left( \tr |\rho|^\alpha \right)^{1/\alpha}$
(the operator norm $\|\rho\|_\infty  =\|\rho\|$ is obtained in the
limit $\alpha\rightarrow\infty$) is related to the $\alpha$-entropy by
\[
  S_\alpha(\rho) = \frac{\alpha}{1-\alpha} \log \|\rho\|_\alpha,
     \qquad
  S_\infty(\rho) = - \log \|\rho\|_\infty,
\]
and satisfies
\[
  \| \rho_A \|_\a + \| \rho_B \|_\a \leq 1 + \| \rho_{AB} \|_\a
\]
for arbitrary bipartite state $\rho_{AB}$.

The following is a strengthening of Audenaert's inequality.
\begin{proposition}
  \label{prop:Aud+}
  Let $\rho_{AB}$ be a bipartite state and $1 < \alpha \leq \infty$. Define
  \begin{equation*}
    M_\a:=\max\left\{\left(\norm{\rho_A}_{\infty}/\norm{\rho_A}_\alpha\right)^{\a-1},
                     \left(\norm{\rho_B}_{\infty}/\norm{\rho_B}_\alpha\right)^{\alpha-1}
              \right\} \quad \text{for} \quad \a<\infty,
  \end{equation*}
  as well as $M_{\infty}:=\lim_{\a\to+\infty}M_\a=\max\{1/m_A,1/m_B\}$, where
  $m_A$ and $m_B$ are the multiplicity of
  $\norm{\rho_A}_{\infty}$ and $\norm{\rho_B}_{\infty}$ as an eigenvalue of
  $\rho_A$ and $\rho_B$, respectively.

  Then,
  \begin{align}
    \norm{\rho_A}_\a+\norm{\rho_B}_\a &\leq
             \min\left\{\kappa+\frac{1}{\kappa}\norm{\rho_{AB}}_\a\,:\,M_\a\le\kappa\right\}
                                                                             \label{Audenaert+} \\
               &=\begin{cases}
                  2\sqrt{\norm{\rho_{AB}}_\a}            &\text{ if } M_\a\le\sqrt{\norm{\rho_{AB}}_\a} \\
                  M_\a+\frac{1}{M_\a}\norm{\rho_{AB}}_\a &\text{ if } \sqrt{\norm{\rho_{AB}}_\a} \leq M_\a
                 \end{cases}
                    \label{Audenaert+2} \\
               &\leq 1+\norm{\rho_{AB}}_\a. \label{Audenaert+3}
  \end{align}
  The last inequality holds with equality if and only if at least one of $\rho_A,\rho_B$ or $\rho_{AB}$ is a pure state. Moreover, we have
  $\norm{\rho_A}_\a+\norm{\rho_B}_\a=1+\norm{\rho_{AB}}_\a$ if and only if
  $\rho_A$ or $\rho_B$ is a pure state.
\end{proposition}
\begin{proof}
We follow Audenaert's proof~\cite{Audenaert:inequality} with a slight modification.
Let $\rho_A=\sum_{i}\lambda_i\pr{e_i}$ and $\rho_B=\sum_{j}\eta_j\pr{f_j}$ be eigen-decompositions
such that the $\lambda$'s and the $\eta$'s are arranged in a decreasing
order. For $\a=\infty$, define $X:=\sum_{i=1}^{m_A}\frac{1}{m_A}\pr{e_i}$,
$Y:=\sum_{i=1}^{m_B}\frac{1}{m_B}\pr{f_i}$, and let $\b:=1$.
For $\a<\infty$, define $X:=\sum_i x_i\pr{e_i}$ and $Y:=\sum_j y_j\pr{f_j}$,
where
\begin{equation*}
  x_i := \lambda_i^{\a-1}/\norm{\rho_A}_\a^{\a-1}
             \quad\text{and}\quad
  y_j := \eta_j^{\a-1}/\norm{\rho_B}_\a^{\a-1},
\end{equation*}
and let $\b$ be such that $\frac{1}{\a}+\frac{1}{\b}=1$.
Let $x$ and $y$ be the vectors formed of the $x_i$'s  and $y_j$'s, respectively.
Then $\norm{X}_\b=\norm{Y}_\b=\norm{x}_\b=\norm{y}_\b=1$ and $\norm{\rho_A}_\a=\Tr X\rho_A$ and
$\norm{\rho_B}_\a=\Tr Y\rho_B$. Hence we have, for any real number $\kappa$, that
\begin{equation}\begin{split}
  \norm{\rho_A}_\a+\norm{\rho_B}_\a
      &= \Tr (X\otimes \1_B+\1_A\otimes Y)\rho_{AB} \\
      &= \kappa+\Tr (X\otimes \1_B+\1_A\otimes Y-\kappa \1_A\otimes \1_B)\rho_{AB} \\
      &\le \kappa+\Tr (X\otimes \1_B+\1_A\otimes Y-\kappa \1_A\otimes \1_B)_+\rho_{AB} \\
      &= \kappa+\Tr Z_{\kappa}\rho_{AB},
  \label{upper bound1}
\end{split}\end{equation}
where $Z_{\kappa}:=(X\otimes \1_B+\1_A\otimes Y-\kappa \1_A\otimes \1_B)_+$
is the positive part.

Consider now the function $a\mapsto f_{\kappa}(a):=\left(\sum_j(y_j+a-\kappa)_+^{\beta}\right)^{\frac{1}{\beta}}=\norm{(y+a-\kappa)_+}_\b$.
This function is convex, $f_{\kappa}(\kappa)=\norm{y}_\b=1$,
and $f_{\kappa}(0)=0$ if we assume that
$\kappa\ge\max_j y_j=\norm{y}_{\infty}$.
Hence, under this assumption, $f_{\kappa}(a)\le a/\kappa$ for every $0\le a\le\kappa$. Thus, if
$\kappa\ge\norm{x}_{\infty}$ then
\begin{align*}
\norm{Z_\kappa}_\b^\beta&=
\norm{(X\otimes \1_B+\1_A\otimes Y-\kappa \1_A\otimes \1_B)_+}_\b^\beta=
\sum_{i,j}(x_i+y_j-\kappa)_+^\beta=\sum_i f_{\kappa}(x_i)^\beta\\
&\le
\sum_i (x_i/\kappa)^\beta=\norm{x}_\b^\beta/\kappa^\beta=1/\kappa^\beta,
\end{align*}
i.e., $\norm{Z_{\kappa}}_{\beta}\le 1/\kappa$. Due to H\"older's inequality,
$\Tr Z_{\kappa}\rho_{AB} \leq \norm{Z_{\kappa}}_\b\norm{\rho_{AB}}_\a \leq \norm{\rho_{AB}}_\a/\kappa$.
Combined with \eqref{upper bound1}, this yields
\begin{equation}
  \label{Aud kappa}
  \norm{\rho_A}_\a+\norm{\rho_B}_\a \leq \kappa+\frac{1}{\kappa}\norm{\rho_{AB}}_\a
                                      =: g(\kappa).
\end{equation}
Since this is true for every $\kappa\ge \max\{\norm{x}_{\infty},\norm{y}_{\infty}\}=M_\a$,
we have proved eq.~\eqref{Audenaert+}.

It is easy to see that $g$ is strictly convex and it has a global minimum
at $\sqrt{\norm{\rho_{AB}}_\a}\le 1$ with minimum value
of $2\sqrt{\norm{\rho_{AB}}_\a}$.
In particular, $g$ is strictly decreasing on the interval $\left(0,\sqrt{\norm{\rho_{AB}}_\a}\right]$
and strictly increasing on $\left[\sqrt{\norm{\rho_{AB}}_\a},1\right]$, and hence we obtain
eq.~\eqref{Audenaert+2}. The inequality (\ref{Audenaert+3}) is obvious.

By the above properties of $g$, we have equality in \eqref{Audenaert+3} if and only if
$\max\left\{M_\a,\sqrt{\norm{\rho_{AB}}_\a}\right\}=1$. 
Obviously, $\sqrt{\norm{\rho_{AB}}_\a}=1$ if and only if $\rho_{AB}$ is a pure state, and it is easy to see that
$M_\a=1$ if and only if $\rho_A$ or $\rho_B$ is pure.

If $\rho_A$ is a pure state then $\rho_{AB}=\rho_A\otimes\rho_B$ and 
$1+\norm{\rho_{AB}}_\a=
1+\norm{\rho_A}_\a\norm{\rho_B}_\a
=
1+\norm{\rho_B}_\a=
\norm{\rho_A}_\a+\norm{\rho_B}_\a$, and a completely similar argument works if $\rho_B$ is pure. On the other hand, 
if $\norm{\rho_A}_\a+\norm{\rho_B}_\a=1+\norm{\rho_{AB}}_\a$ then equality has to hold in \eqref{Audenaert+3},
and hence $\rho_A,\rho_B$ or $\rho_{AB}$ has to be pure. If $\rho_{AB}$ is pure but $\rho_A$ is not then
$\norm{\rho_{A}}_\a=\norm{\rho_{B}}_\a<1$  and hence 
$\norm{\rho_A}_\a+\norm{\rho_B}_\a<2=  1+\norm{\rho_{AB}}_\a$. This proves the last assertion about the equality case.
\end{proof}

\medskip
It appears that even for two parties, no description of $\Sigma_\alpha^2$
or $\overline{\Sigma_\alpha^2}$ is known. Nor, which other inequalities
there are constraining the latter.

\section{Classical case}
\label{sec:classical}
As remarked in the introduction, if restricted to classical states
$\rho$, the R\'enyi entropies are monotonic, i.e.
\begin{equation}
  \label{eq:monotonic}
  I \subset J \ \Rightarrow\ S_\alpha(\rho_I) \leq S_\alpha(\rho_J).
\end{equation}
(More generally, this holds for separable states, thanks to the
majorisation result of Kempe and Nielsen~\cite{KempeNielsen:sep}.)
In this section, we denote the set of $\alpha$-entropic vectors of
a generic distribution of $n$ random variables by $\Gamma_\alpha^n$.
In other words, this is a subset of $\Sigma_\alpha^n$, with the
restriction that the underlying states are classical.

The extremal rays of the convex cone $\cM\cO^n$ described by non-negativity and
eqs.~(\ref{eq:monotonic}) -- which thus contains $\Gamma_\alpha^n$ --
are easy to describe in combinatorial language~\cite{ILW:relent}:
They are precisely the rays spanned by the indicator functions
\[
  \dot\iota_{\cU}:I \longmapsto \begin{cases}
                                  1 & \text{ if } I \in \cU, \\
                                  0 & \text{ otherwise}.
                                \end{cases}
\]
of a nonempty set family $\cU \subset \powersetnonempty[n]$ with the property
that $J \supset I\in \cU$ implies $J\in \cU$ (hence always $[n] \in \cU$).
Such set families are known in combinatorics as ``upsets'' (or sometimes ``ideals'').

Some of the simplest upset are generated by a single element:
\[
  \uparrow\!J = \{ I \in \powersetnonempty[n] : J \subset I \}.
\]
These have the property that the unique minimal element of the
family is $J$. Note also that an upset contains, with each element
$J$, the entire $\uparrow\!J$. This means that every upset $\cU$
can be written
\[
  \cU = \bigcup_{J \in \cL} \uparrow\!J,
\]
with $\cL$ the set of minimal elements of $\cU$.

\medskip
For instance for $n=2$, there are four upsets and clearly all 
four associated rays
are attainable (whole ray for $\alpha < 1$, sufficiently long dilution
for $\alpha > 1$).

Next we show that this is the only difference to the quantum case,
at least as long as we are only looking for homogeneous inequalities.
Namely, the only homogeneous inequalities obeyed by the classical
$\alpha$-entropies are non-negativity and monotonicity.

\begin{theorem}
  \label{thm:0-alpha-1:c}
  Let $0 < \alpha < 1$.
  For any upset $\cU \subset \powersetnonempty[n]$ and all $s > 0$, there is
  a vector $s\,\dot\iota_{\cU} + O(1) \in \Gamma_\alpha^n$.
  To be precise, there exists a probability distribution $P$ with
  \[
    s \leq H_\a(P_I)\le s+\frac{\log|\cL|+1}{1-\alpha}\ds  \text{ for }\ds I\in\cU, \qquad
    H_\a(P_I) \leq \log|\cL|+1 \ds \text{ for }\ds I\not\in\cU.
  \]
  In particular, for $s \rar \infty$ we obtain
  $\dot\iota_{\cU} \in \overline{\RR_{\geq 0}\,\Gamma_\alpha^n}$.
  As a consequence, $\overline{\RR_{\geq 0}\,\Gamma_\alpha^n} = \cM\cO^n$.
\end{theorem}
\begin{proof}
  Let $\cU = \bigcup_{J\in\cL} \uparrow\!J$, which can be achieved
  by choosing $\cL$ to be the minimal elements of $\cU$.
For each $i\in[n]$, let the local alphabet $\X_i$ be of the form
\begin{equation*}
\X_i=\cup_{J\in\cL}^*\X_{i,J},\ds\text{where}\ds
\begin{cases}
\X_{i,J}=\{0\}\cup[M_i],&\text{if}\ds i\in J,\\
|\X_{i,J}|=1&\text{if}\ds i\notin J,
\end{cases}
\end{equation*}
where $\cup^*$ denotes disjoint union, i.e.,
$\X_{i,J}$ and $\X_{i,J'}$ are disjoint if $J\ne J'$. For each $J\in\cL$, let 
$P_{t_J;\{M_i:i\in J\}}$ be a probability distribution on 
$\times_{i\in J}\X_{i,J}\subset\times_{i\in J}\X_i$, defined as in 
lemma~\ref{lemma:0-alpha-1}, with $t_J^\alpha M_{J}^{1-\alpha} = 2^{s'(1-\alpha)}$, 
where $s':=s+\frac{\alpha}{1-\alpha}\log|\cL|$.
Let $Q_J:=P_{t_J;\{M_i:i\in J\}}\otimes\delta_{J^c}$, where 
$\delta_{J^c}$ is the trivial probability distribution on the
single-element set $\times_{i\in J^c}\X_{i,J}$. 
Note that the supports of $Q_J$ and $Q_{J'}$ are disjoint for $J\ne J'$.
We claim that 
\begin{equation*}
P :=\frac{1}{|\cL|} \sum_{J\in\cL} Q_J= \frac{1}{|\cL|} \sum_{J\in\cL} P_{t_J;\{M_i:i\in J\}} \otimes \delta_{{J^c}}
\end{equation*}
has the desired properties for large enough $M_0 := \min_i M_i$.

Indeed, it is easy to see that for any $\emptyset\ne I\subset[n]$, 
\begin{align}
H_\alpha(P_I)&=
\frac{1}{1-\alpha}\log\bz
\frac{1}{|\cL|^\alpha}\bz
\sum_{I\cap J=\emptyset}1+\sum_{J\subset I}\bz(1-t_J)^\alpha+2^{s(1-\alpha)}\jz
+\sum_{\emptyset\ne J\setminus I\subsetneq J}\bz(1-t_J)^\alpha+2^{s(1-\alpha)}/M_{J\setminus I}^{1-\alpha}\jz
\jz\jz\nonumber\\
&=
\frac{1}{1-\alpha}\log\bz
\frac{1}{|\cL|^\alpha}\bz
\sum_{I\cap J=\emptyset}1+
\sum_{I\cap J\ne\emptyset}(1-t_J)^\alpha\jz
+\frac{1}{|\cL|^\alpha}\sum_{J\subset I}2^{s(1-\alpha)}
+\frac{1}{|\cL|^\alpha}\sum_{\emptyset\ne J\setminus I\subsetneq J}2^{s(1-\alpha)}/M_{J\setminus I}^{1-\alpha}
\jz\nonumber\\
&\xrightarrow[M_0\to+\infty]{}
\frac{1}{1-\alpha}\log\bz
|\cL|^{1-\alpha}+2^{s'(1-\alpha)}\frac{\#\{J\in\cL\,:\,J\subset I\}}{|\cL|^\alpha}
\jz,\nonumber\\
&=
\frac{1}{1-\alpha}\log\bz
|\cL|^{1-\alpha}+2^{s(1-\alpha)}\#\{J\in\cL\,:\,J\subset I\}
\jz,
\label{classical alpha 0-1}
\end{align}
where in all summations over $J$ above, $J\in\cL$ is implicit.
Now, if $I\in \cU$ then there exists a $J\in\cL$ such that $J\subset I$, and hence
the expression in \eqref{classical alpha 0-1} can be lower bounded by $s$, and upper bounded by 
$\frac{1}{1-\alpha}\log\bz|\cL|\bz 1+2^{s(1-\alpha)}\jz\jz< s+\frac{\log|\cL|+1}{1-\alpha}$.
On the other hand, if $I\notin \cU$ then $\#\{J\in\cL\,:\,J\subset I\}=0$ and hence
the expression in \eqref{classical alpha 0-1} is equal to $\log|\cL|$.
\end{proof}

\begin{theorem}
  \label{thm:1-alpha-infty:c}
  Let $1 < \alpha \leq \infty$.
  For any upset $\cU \subset \powersetnonempty[n]$ and $s > 0$, there is
  a vector $s\,\dot\iota_{\cU} + O(1) \in \Gamma_\alpha^n$.
  To be precise, there exists a constant $C$, which may be chosen
  as 
$C=\frac{1}{1-\frac{1}{\alpha}}\log(2n^k)$
  if $\cU$ is generated by $k$ elements ($k < 2^n$ always),
  and a probability distribution $Q$ with
  \[
    s \leq H_\a(Q_I) \leq s + C \text{ for } I\in\cU, \qquad
    H_\a(Q_I) \leq C \ \text{ for } I\not\in\cU.
  \]
  In particular, for $s \rar \infty$ we obtain
  $\dot\iota_{\cU} \in \overline{\RR_{\geq 0}\,\Gamma_\alpha^n}$.
  As a consequence, $\overline{\RR_{\geq 0}\,\Gamma_\alpha^n} = \cM\cO^n$.
\end{theorem}
\begin{proof}
Let $M$ be the smallest natural number such that 
$s\le 1+\log M$.
  We take as our building blocks the distributions $Q_{[n]:M}$ from
  lemma~\ref{lemma:1-alpha-infty} and its proof, for simplicity with
  uniform $R$ on $[M]$. Furthermore, define
  the following uniform distribution on the diagonal of $[M]^n$:
  \[
    \Delta_M := \frac{1}{M} \sum_{x=1}^M \delta_x\otimes\delta_x\otimes\cdots\otimes\delta_x.
  \]

  Now, for upset $\cU = \bigcup_{J\in\cL} \uparrow\!J$, let
  \[
    Q := \frac{1}{2} \Delta_M \oplus
         \frac{1}{2} \bigotimes_{J\in\cL} \bigl( Q_{J:M} \otimes \delta_{J^c} \bigr),
  \]
  where $\delta_{J^c}$ as before refers to a generic point mass for
  parties $J^c = [n]\setminus J$, the product over $J\in\cL$ implies
  Cartesian products of the local alphabets, and the direct sum
  likewise a direct sum of the \emph{local} alphabets.

  The first term in the direct sum makes sure that in each
  marginal the largest probability value occurring is at least $\geq\frac{1}{2M}$,
  with multiplicity $M$, not allowing the R\'enyi entropy of
  any subset to become larger than $\log M + \frac{1}{1-\frac{1}{\a}}$.
  Turning to the second term, note that in the tensor product over $J\in\cL$,
  the distributions are designed such that for $J\subset I$,
  the distribution $\bigl( Q_{J:M} \otimes \delta_{J^c} \bigr)_I$ is uniform on
  $|J|M$ elements, whereas for $J\not\subset I$, it has at least
  one value $\geq\frac{1}{|J|}$.

  Thus,
  \[
    H_\a(Q_I) \ \begin{cases}
                  \leq \frac{1}{1-\frac{1}{\a}}\log(2n^k) & \text{ for } I \not\in \cU, \\
                  \geq 1+\log M                           & \text{ for } I \in \cU,
                \end{cases}
  \]
  and we are done.
\end{proof}

\section{Epilogue}
\label{sec:epi}
We have carried out an analysis of the inequalities obeyed by 
quantum R\'{e}nyi entropies in multi-partite systems, 
in analogy to the very deep ongoing programme for the von Neumann entropy. 
In the quantum case, our findings can be summarized concisely as saying 
that apart from trivial non-negativity of individual entropies there
are no inequalities obeyed by the R\'{e}nyi $\alpha$-entropies of a
multipartite state, when $0 < \alpha < 1$. For $1 < \alpha \leq \infty$ there are no
other homogeneous inequalities, but the set of attainable entropic
vectors is not a cone, meaning that there are further, non-homogeneous
inequalities. In the classical case (and more broadly that of separable
quantum states) there is furthermore monotonicity in the sense that
a smaller subset of parties cannot have larger entropy, and we could
show similarly that this is the only homogeneous inequality for
all $\alpha \neq 1$. It is curious to contrast this with the limit
$\alpha = 1$, the von Neumann entropy, which is subject to subadditivity
and strong subadditivity, as well as triangle inequality and weak
monotonicity, all crucial relations for the development of statistical
mechanics and quantum information theory. The classical case has
even more inequalities, due to Zhang and Yeung and subsequent work.

We did not discuss the other limit $\alpha = 0$, for which the R\'{e}nyi
entropy is the logarithm of the rank of the density operator, which
indeed behaves rather differently from the other $\alpha$-entropies:
For one thing, it takes only discrete values in the logarithm of integers,
and it is discontinuous. Furthermore, it is easy to see that it obeys
subadditivity
\[
  S_0(\rho_{I \cup J}) \leq S_0(\rho_I) + S_0(\rho_J),
\]
and it is unknown which other  inequalities (whether homogeneous, linear or other) it
 satisfies. Note however that it definitely does not satisfy strong subadditivity
 \cite{Cadney:personal-communication-june2012}.

We leave a few other open questions and further directions for future
investigations: For instance, we would like to know all necessary
non-homogeneous inequalities to describe the $\overline{\Sigma^n_\alpha}$.
Note that the classical/separable sets $\Gamma^n_\alpha$ for $\alpha
> 1$ are not cones either, but what about $0 < \alpha < 1$? Finally,
and perhaps most interestingly in the light of recovering what rich
structure is known for $\alpha = 1$, can we extend the present investigation
to relations between R\'{e}nyi entropies for different $\alpha$? For
example, it is well-known that for $\alpha < \beta$,
$S_\alpha(\rho) \geq S_\beta(\rho)$, or $S_\alpha(\rho_{AB}) \leq
S_\alpha(\rho_B) + S_0(\rho_B)$, but we do not know which other
inequalities, if any, exist.

\acknowledgments
We thank Marcus Huber and Josh Cadney for conversations on R\'{e}nyi entropies,
and especially Josh Cadney for providing a counterexample showing that the 
$0$-entropy does not obey strong subadditivity.

NL and AW are supported by the European Commission (STREP ``QCS'' and
IP ``QESSENCE''). 
MM is supported by a Marie Curie Fellowship (``QUANTSTAT'') of the European
Commission.
AW received furthermore support from the ERC (Advanced Grant ``IRQUAT''),
a Royal Society Wolfson Merit Award and a Philip Leverhulme Prize.
The Centre for Quantum Technologies is funded by the Singapore
Ministry of Education and the National Research Foundation as part
of the Research Centres of Excellence programme.

\end{document}